\documentclass[a4paper, 11pt]{amsart}
\usepackage{amsmath, amssymb, amsthm}	
\usepackage{braket,mathrsfs}
\usepackage{cite}
\usepackage[all]{xy}
\usepackage{youngtab}					
\usepackage[draft]{graphicx}

\calclayout

\newcommand{\pr}[1]{#1^{\prime}}

\newcommand{\pppr}[1]{#1^{\prime\prime\prime}}

\newcommand{\del}{\partial}

\newcommand{\mfrak}[1]{\mathfrak{#1}}
\newcommand{\mcal}[1]{\mathcal{#1}}
\newcommand{\mbb}[1]{\mathbb{#1}}
\newcommand{\mrm}[1]{\mathrm{#1}}

\theoremstyle{plain}
\newtheorem{thm}{Theorem}[section]

\theoremstyle{definition}
\theoremstyle{remark}

\makeatletter
    
    \@addtoreset{equation}{section}
\makeatother

\title{Note on Schramm-Loewner evolution for superconformal algebras}
\author{Shinji Koshida}
\address{Department of Basic Science, The University of Tokyo}
\email{koshida@vortex.c.u-tokyo.ac.jp}

\begin{document}

\begin{abstract}
We propose variants of Schramm-Loewner evolution (SLE) that are related to superconformal algebras
following the group theoretical formulation of SLE,
in which the relevant stochastic differential equation is derived from a random process
on an infinite dimensional Lie group.
In this paper, we consider random processes on certain kind of groups of superconformal transformations
generated by exponentiated elements of the Grassmann envelop of superconformal algebras.
We also provide a prescription of obtaining local martingales from a representation of
the superconformal algebra after integration by Grassmann variables.

\end{abstract}

\maketitle

\section{Introduction}
The interplay between Schramm-Loewner evolution (SLE) \cite{Schramm2000} and 
conformal field theory (CFT) \cite{BelavinPolyakovZamolodchikov1984} in two dimensions has been explored from various points of view,
and deep connections between them have now been established with the name of SLE/CFT correspondence.
In the most na\"{i}ve sense, SLE/CFT correspondence allows one to compute local martingales associated with SLE
from a degenerate representation of the Virasoro algebra due to the correspondence between
the stochastic differential equation for a random process on the group of conformal transformations
that generates SLE and the singular vector in the representation \cite{BauerBernard2003a}.
The essence of this explanation is to identify SLE, which describes deformation of simply connected domains,
with a random process on an infinite dimensional Lie group that acts on a representation of the Virasoro algebra.
This framework is often called the group theoretical formulation of SLE.
In a slightly different approach to SLE/CFT correspondence, one regards SLE as a random process on a moduli space of Riemann surfaces
and finds the probability measure be a section of the determinant bundle of the moduli space \cite{FriedrichKalkkinen2004,Kontsevich2003}.
This approach can be locally seen as the group theoretical formulation via the Virasoro uniformization,
but further allows one to consider a generalization of SLE on other Riemann surfaces than a simply connected one.
The construction of SLE measure on a path space was carried in \cite{KontsevichSuhov2007} and \cite{Dubedat2015a,Dubedat2015b} in slightly different formulations,
and it was found that the partition function is a section of a line bundle on a Teichm\"{u}ller space \cite{Dubedat2015a,Dubedat2015b}.

Following SLE/CFT correspondence, several generalizations of SLE corresponding to other CFTs than
ones associated with the Virasoro algebra have been proposed,
examples of which include multiple SLE \cite{BauerBernardKytola2005}
and SLE with internal degrees of freedom \cite{BettelheimGruzbergLudwigWiegmann2005,AlekseevBytskoIzyurov2011, SK2017, SK2018a}
that corresponds to Wess-Zumino-Witten theory and its super analogue \cite{SK2018b}.
A generalized SLE we study in this paper is one on a super Riemann surface, once considered in \cite{Rasmussen2004a,NagiRasmussen2005} for $\mcal{N}=1$ case.
In their paper, the authors proposed a random process associated with the most simple nontrivial singular vector
in a representation of the $\mcal{N}=1$ superconformal algebra, which is an odd one.
As is mentioned above, SLE/CFT correspondence should allow one to obtain local martingales from a representation.
In case that supersymmetry is involved, the the stochastic differential equation for a random process on a Lie group
does not necessarily correspond to a singular vector in a representation due to Grassmann variables,
which prevents one from consideridering even singular vectors.
We address this issue by proposing a prescription of obtaining local martingales after a certain integral by Grassmann variables
following our previous work \cite{SK2018b}, and construct a generalized SLE on a super Riemann surface.

This paper is organized as follows.
The next Sect.\ref{sect:group_theor} gives a brief overview of the group theoretical formulation of SLE
corresponding to the Virasoro algebra.
Sect.\ref{sect:superconformal} serves as a preliminary part on the theory of superconformal transformation and super Riemann surfaces.
In Sect.\ref{sect:N=1}, we extend the formulation in Sect.\ref{sect:group_theor} to
the Neveu-Schwarz sector of the $\mcal{N}=1$ superconformal algebra,
and derive stochastic differential equations that can be interpreted as generalization of SLE.
We also present a way of obtaining local martingales from a representation of
the superconformal algebra after integrating Grassmann variables.
In Sect.\ref{sect:N=2}, we generalize the content of Sect.\ref{sect:N=1} to the case of $\mcal{N}=2$.
In Sect.\ref{sect:discussion}, we make some discussion on our result and future directions.

\section{Group theoretical formulation of SLE}
\label{sect:group_theor}

In this section, we recall the group theoretical formulation of SLE presented in \cite{BauerBernard2003a}.
SLE is a one-parameter family $\{g_{t}(z)\in \mbb{C}[[z^{-1}]]z\}_{t \ge 0}$ of formal power series
that satisfies the stochastic differential equation
\begin{equation}
	\frac{d}{dt}g_{t}(z)=\frac{2}{g_{t}-\sqrt{\kappa}B_{t}},\ \ \ g_{0}(z)=z.
\end{equation}
Here $B_{t}$ is the standard Brownian motion on $\mbb{R}$ that start from the origin
and $\kappa>0$ is a parameter.
SLE specified by the parameter $\kappa$ is often denoted by SLE$(\kappa)$.
The formal power series at each time $t$ becomes a uniformization map of a simply connected domain.
Namely, there exists a subset $K_{t}$, called a hull, of the upper half plane $\mbb{H}=\{z\in\mbb{C}|\mrm{Im}z>0\}$ so that
$g_{t}$ is a biholomorphic function $g_{t}:\mbb{H}\backslash K_{t}\to \mbb{H}$.
The evolution of hulls $\{K_{t}\}_{t\ge 0}$ is known to be increasing, {\it i.e.},
$t<s$ implies $K_{t}\subset K_{s}$.
Thus SLE describes a growth process of hulls in the upper half plane,
which is closely related to cluster interfaces in a two dimensional critical system.

To understand SLE in the group theoretical formulation, it is convenient to use
an alternative presentation of SLE.
Let us set $f_{t}(z)=g_{t}(z)-\sqrt{\kappa}B_{t}$.
Then we have the following stochastic differential equation
\begin{equation}
	\label{eq:SLE_alter}
	df_{t}(z)=\frac{2dt}{f_{t}(z)}-\sqrt{\kappa}dB_{t},\ \ \ f_{0}(z)=z,
\end{equation}
which can be connected to representation theory of the Virasoro algebra.

The first step of the group theoretical formulation is
to interpret SLE in Eq.(\ref{eq:SLE_alter}) as formal coordinate transformation at infinity.
Following the terminology of \cite{FrenkelBen-Zvi2004}, let $\mcal{O}=\mbb{C}[[w]]$ be a completed topological $\mbb{C}$-algebra
\footnote{
More intrinsically, for a point $X$ of a Riemann surface,
one has the stalk $\mcal{O}_{X}$ of the structure sheaf at $X$ and its maximal ideal $\mfrak{m}_{X}$.
Then our algebra $\mcal{O}$ is constructed as $\mcal{O}=\varprojlim_{n} \mcal{O}_{X}/\mfrak{m}_{X}^{n}$.
},
and $\mrm{Aut}\mcal{O}$ be a group of continuous automorphisms of $\mcal{O}$.
An automorphism $\rho\in \mrm{Aut}\mcal{O}$ is identified with its image of the topological generator $w$:
\begin{equation}
	\rho(w)=a_{1}w+a_{2}w^{2}+\cdots.	
\end{equation}
We regard the formal disc $D=\mrm{Spec}\mcal{O}$ as the formal neighborhood at the infinity.
Then the coodinate $z$ at the origin is transformed by the same transformation $\rho$ as
\begin{equation}
	z\mapsto \frac{1}{\rho(1/z)}=b_{1}z+b_{0}+b_{-1}z^{-1}+\cdots.
\end{equation}
Thus we can see that the group $\mrm{Aut}\mcal{O}$ is identified with a set of formal power series
\begin{equation}
	\label{eq:Aut_series}
	\mrm{Aut}\mcal{O}\simeq \{\rho(z)=b_{1}z+b_{0}+b_{-1}z^{-1}+\cdots|b_{1}\neq 0\}.
\end{equation}
A significant subgroup in our context denoted by $\mrm{Aut}_{+}\mcal{O}$ is defined by 
adding the condition that $b_{1}=1$.

The Lie algebra of the group $\mrm{Aut}\mcal{O}$ consists of vector fields
holomorphic at the infinity, which is realized as a subalgebra $\mrm{Der}_{0}\mcal{O}=z\mbb{C}[[z^{-1}]]\del_{z}$
of the Witt algebra.
Similarly, the Lie algebra of the subgroup $\mrm{Aut}_{+}\mcal{O}$ is identified with
$\mrm{Der}_{+}\mcal{O}=\mbb{C}[[z^{-1}]]\del_{z}$.
The exponential map $\mrm{Der}_{0}\to\mrm{Aut}\mcal{O}$ is presented in the following way.
For an automorphism $\rho\in\mrm{Aut}\mcal{O}$, we can uniquely find numbers $v_{i}$ so that
\begin{equation}
	\label{eq:exp_map}
	\exp\left( \sum_{i<0}v_{i}z^{i+1}\del_{z}\right) v_{0}^{z\del_{z}}\cdot z=\rho(z).
\end{equation}

From the normalization of $f_{t}(z)$ at the infinity, SLE in Eq.(\ref{eq:SLE_alter}) can be regarded
as a random process on the infinite dimensional Lie group $\mrm{Aut}_{+}\mcal{O}$ under the identification
in Eq.(\ref{eq:Aut_series}).
It can be also verified that Eq.(\ref{eq:SLE_alter}) for $f_{t}(z)$ is equivalent to
the following stochastic differential equation for $f_{t}$ considered as a random process on
$\mrm{Aut}_{+}\mcal{O}$:
\begin{equation}
	\label{eq:SLE_vector_field}
	f_{t}^{-1}df_{t}=\left(-2\ell_{-2}+\frac{\kappa}{2}\ell_{-1}^{2}\right)dt+\sqrt{\kappa}\ell_{-1}dB_{t},
\end{equation}
where we set $\ell_{n}=-z^{n+1}\del_{z}$.

We next construct a representation of the group $\mrm{Aut}\mcal{O}$ on the formal completion of
a representation of the Virasoro algebra.
The Virasoro algebra $\mrm{Vir}$ is an infinite dimensional Lie algebra
$\mrm{Vir}=\bigoplus_{n\in\mbb{Z}}\mbb{C}L_{n}\oplus\mbb{C}C$ with Lie bracket
\begin{equation}
	[L_{m},L_{n}]=(m-n)L_{m+n}+\frac{m^{3}-m}{12}\delta_{m+n,0}C,\ \ \ [C,\mrm{Vir}]=\{0\}.
\end{equation}
Its highest weight representations are classified by central charge $c$ and conformal weight $h$.
Namely, the highest weight vector denoted by $\ket{c,h}$ is an eigenvector of $C$ and $L_{0}$
with eigenvalues $c$ and $h$, respectively, and annihilated by $L_{n}$ for $n>0$.
Then the corresponding irreducible highest weight representation $L(c,h)$
is the irreducible quotient of the Verma module induced from $\ket{c,h}$.
Each irreducible representation $L(c,h)$ decomposes into direct sum of 
eigenspaces of $L_{0}$ so that $L(c,h)=\bigoplus_{n=0}^{\infty}L(c,h)_{h+n}$,
where $L(c,h)_{\Delta}=\{v\in L(c,h)|L_{0}v=\Delta v\}$ is the eigenspace
of $L_{0}$ corresponding to an eigenvalue $\Delta$.
Then the formal completion of $L(c,h)$ is
$\overline{L(c,h)}=\prod_{n=0}^{\infty}L(c,h)_{h+n}$.

We firstly equip the space $\overline{L(c,h)}$ for an integer $h$ with an action of the group $\mrm{Aut}\mcal{O}$.
Under the identification in Eq.(\ref{eq:Aut_series}), we find numbers $v_{i}$ for $i\le 0$ for an automorphism $\rho$
so that Eq.(\ref{eq:exp_map}) holds.
Then the operator
\begin{equation}
	Q(\rho)=\exp\left(-\sum_{i<0}v_{i}L_{i}\right)v_{0}^{-L_{0}}
\end{equation}
defines a representation of $\mrm{Aut}\mcal{O}$ on $\overline{L(c,h)}$.
Here $v_{0}^{-L_{0}}$ acts as multiplication by $v_{0}^{-\Delta}$ on each eigenspace $L(c,h)_{\Delta}$,
which is well-defined if $h$ is an integer.
If $h$ is not an integer, the whole group $\mrm{Aut}\mcal{O}$ cannot act, 
but the subgroup $\mrm{Aut}_{+}\mcal{O}$ can since $v_{0}=1$ for an automorphism in this subgroup
and we do not encounter the branch issue for the part $v_{0}^{-L_{0}}$.

As we have seen above, SLE as a random process on $\mrm{Aut}_{+}\mcal{O}$ is characterized by Eq.(\ref{eq:SLE_vector_field}).
Combining the representation $Q$ of $\mrm{Aut}_{+}\mcal{O}$,
we have an operator-valued random process $Q(f_{t})$,
which satisfies the stochastic differential equation
\begin{equation}
	Q(f_{t})^{-1}dQ(f_{t})=\left(-2L_{-2}+\frac{\kappa}{2}L_{-1}^{2}\right)dt +\sqrt{\kappa}L_{-1}dB_{t}.
\end{equation}
Notice that for a given parameter $\kappa$, we have
$\left(-2L_{-2}+\frac{\kappa}{2}L_{-1}^{2}\right)\ket{c_{\kappa},h_{\kappa}}=0$
for $c_{\kappa}=1-\frac{3(\kappa-4)^{2}}{2\kappa}$ and $h_{\kappa}=\frac{6-\kappa}{2\kappa}$.
This implies that the representation-space-valued random process $Q(f_{t})\ket{c_{\kappa},h_{\kappa}}$ is a local martingale,
which generates infinitely many local martingales associated with SLE.
An example of such obtained local martingales is
\begin{equation}
	\braket{c_{\kappa},h_{\kappa}|L(z)Q(f_{t})|c_{\kappa},h_{\kappa}}
	=h_{\kappa}\left(\frac{\pr{f}_{t}(z)}{f_{t}(z)}\right)^{2}+\frac{c_{\kappa}}{12}(Sf_{t})(z),
\end{equation}
where $L(z)=\sum_{n\in\mbb{Z}}L_{n}z^{-n-2}$ is the Virasoro field (the stress-energy tensor, in other words)
and $(S\rho)(z)=\frac{\pppr{\rho}(z)}{\pr{\rho}(z)}-\frac{3}{2}\left(\frac{\pr{\rho}(z)}{\rho(z)}\right)^{2}$
is the Schwarzian derivative.
While the quantity in the right hand side is checked to be a local martingale
by a direct calculation, the left hand side reveals its CFT origin.

\section{Superanalytic functions and superconformal maps}
\label{sect:superconformal}
In this section, we recall the notion of superanalytic functions and superconformal maps
following literatures \cite{Barron1996,Barron2003,Barron2007}
that is needed in the construction of  SLE with supersymmetry.

\subsection{Superanalytic superfunction}
Let $\bigwedge_{N}=\bigwedge [\zeta_{1},\cdots,\zeta_{N}]$ be the Grassmann algebra over the field of complex numbers $\mbb{C}$
that is generated by $N$ variables $\zeta_{1},\cdots,\zeta_{N}$.
We have a natural inclusion $\bigwedge_{N}\hookrightarrow \bigwedge_{M}$ for $N<M$
and the direct limit of this inductive system is denoted by $\bigwedge_{\infty}$.
We denote a Grassmann algebra by $\bigwedge_{\ast}$ avoiding to specify the number of generators $N$ or $\infty$.
We introduce sets of indices $I_{\ast}=\{(i)=(i_{1}<\cdots<i_{2n})|i_{k}=1,2,\cdots, k=1,\cdots,2n, n=1,2,\cdots\}\cup \{(\emptyset)\}$,
$J_{\ast}=\{(j)=(j_{1}<\cdots<j_{2n-1})|j_{k}=1,2,\cdots, k=1,\cdots,2n-1, n=1,2,\cdots\}$ and $K_{\ast}=I_{\ast}\cup J_{\ast}$.
For each $(i)=(i_{1}<\cdots<i_{n})\in K_{\ast}$, we set $\zeta_{(i)}=\zeta_{i_{1}}\cdots\zeta_{i_{n}}$
and $\zeta_{(\emptyset)}=1$.
The Grassmann algebra $\bigwedge_{\ast}$ naturally admits a $\mbb{Z}_{2}$-gradation specified by
$\bigwedge_{\ast}^{0}=\mrm{Span}\{\zeta_{(i)}|(i)\in I_{\ast}\}$ and
$\bigwedge_{\ast}^{1}=\mrm{Span}\{\zeta_{(j)}|(j)\in J_{\ast}\}$.
We also introduce another direct sum decomposition $\bigwedge_{\ast}=(\bigwedge_{\ast})_{\mrm{B}}\oplus(\bigwedge_{\ast})_{\mrm{S}}$,
where $(\bigwedge_{\ast})_{\mrm{B}}=\mbb{C}\zeta_{(\emptyset)}\simeq \mbb{C}$ (body) and
$(\bigwedge_{\ast})_{\mrm{S}}=\mrm{Span}\{\zeta_{(i)}|(i)\in K_{\ast}\backslash \{(\emptyset)\}\}$ (soul).
Each element $a\in \bigwedge_{\ast}$ is written as $a=a_{\mrm{B}}+a_{\mrm{S}}$ along this direct sum decomposition
where $a_{\mrm{B}}\in(\bigwedge_{\ast})_{\mrm{B}}$ and $a_{\mrm{S}}\in(\bigwedge_{\ast})_{\mrm{S}}$.

Let $f$ be an analytic function on a domain in $\mbb{C}^{m}$.
Then for an $m$-tuple $(z_{1},\cdots,z_{m})\in(\bigwedge_{\ast}^{0})^{m}$ of even elements, we define
\begin{equation}
	f(z_{1},\cdots,z_{m})=\sum_{\ell_{1},\cdots,\ell_{m}=0}^{\infty}\frac{(z_{1})_{\mrm{S}}^{\ell_{1}}\cdots (z_{m})_{\mrm{S}}^{\ell_{m}}}{\ell_{1}!\cdots \ell_{m}!}
		\frac{\del^{\ell_{1}+\cdots+\ell_{m}}}{\del (z_{1})_{\mrm{B}}^{\ell_{1}}\cdots \del (z_{m})_{\mrm{B}}^{\ell_{m}}}f((z_{1})_{\mrm{B}},\cdots,(z_{m})_{\mrm{B}}).
\end{equation}
The right hand side reduces to a finite sum, thus this quantity is well-defined if the body $((z_{1})_{\mrm{B}},\cdots,(z_{m})_{\mrm{B}})$
lies in the domain of $f$.
Note also that $f(z_{1},\cdots,z_{m})\in \bigwedge_{\ast}^{0}$ if it makes sense.

We shall define the notion of a superanalytic function on $(\bigwedge_{\ast}^{0})^{m}\oplus (\bigwedge_{\ast}^{1})^{n}$.
Let $\pi_{\mrm{B}}^{(m,n)}:(\bigwedge_{\ast}^{0})^{m}\oplus (\bigwedge_{\ast}^{1})^{n}\to \mbb{C}^{m}$ be the projection taking the body of even parts.
Then we equip $(\bigwedge_{\ast}^{0})^{m}\oplus (\bigwedge_{\ast}^{1})^{n}$ with the pull-back topology of the natural topology on $\mbb{C}^{m}$
via the projection $\pi_{\mrm{B}}^{(m,n)}$, which is called the DeWitt topology.
A function $H$ on an open set $U\subset (\bigwedge_{\ast}^{0})^{m}\oplus (\bigwedge_{\ast}^{1})^{n}$ with values in $\bigwedge_{\ast}$ is said to be
a superanalytic function in $(m,n)$-variables if it has the form
\begin{equation}
	H(z_{1},\cdots,z_{m},\theta_{1},\cdots,\theta_{n})=\sum_{(\ell)\in K_{n}}\theta_{\ell_{1}}\cdots\theta_{\ell_{j}}f_{(\ell)}(z_{1},\cdots,z_{m}),
\end{equation}
where $f_{(\ell)}$ are
\begin{equation}
	f_{(\ell)}(z_{1},\cdots,z_{m})=\sum_{(k)\in K_{\ast}}f_{(\ell),(k)}(z_{1},\cdots,z_{m})\zeta_{(k)}
\end{equation}
with analytic functions $f_{(\ell),(k)}$.
A superanalytic function $H$ is said to be even (resp. odd) if it takes values in the even (resp. odd) part of the Grassmann algebra.

The left partial derivatives by $z_{i}$ and $\theta_{i}$ acting on a superanalytic function $H$ are defined by
\begin{align}
	&\delta z_{i}\frac{\del}{\del z_{i}}H(z_{1},\cdots,z_{m},\theta_{1},\cdots,\theta_{n})+O((\delta z_{i})^{2}) \\
	&=H(z_{1},\cdots,z_{i}+\delta z_{i},\cdots,z_{m},\theta_{1},\cdots,\theta_{n})-H(z_{1},\cdots,z_{m},\theta_{1},\cdots,\theta_{n}) \notag
\end{align}
for arbitrary $\delta z_{i}\in\bigwedge_{\ast}^{0}$ and
\begin{align}
	&\delta \theta_{i}\frac{\del}{\del \theta_{i}}H(z_{1},\cdots,z_{m},\theta_{1},\cdots,\theta_{n})+O((\delta \theta_{i})^{2}) \\
	&=H(z_{1},\cdots,z_{m},\theta_{1},\cdots,\theta_{i}+\delta\theta_{i},\cdots,\theta_{n})-H(z_{1},\cdots,z_{m},\theta_{1},\cdots,\theta_{n}) \notag
\end{align}
for arbitrary $\delta \theta_{i}\in\bigwedge_{\ast}^{1}$.
Then $\del/\del z_{i}$ and $\del/\del\theta_{i}$ define an even and odd operation, respectively.

\subsection{Superconformal map of $\mcal{N}=1$}
We consider in this subsection the case that $(m,n)=(1,1)$.
Let $H$ be a superanalytic map $H:\bigwedge_{\ast}^{0}\oplus \bigwedge_{\ast}^{1}\to \bigwedge_{\ast}^{0}\oplus \bigwedge_{\ast}^{1}$
that sends $(z,\theta)\mapsto (\tilde{z}=H^{0}(z,\theta),\tilde{\theta}=H^{1}(z,\theta))$, {\it i.e.},
it is a pair $(H^{0}, H^{1})$ of an even superanalytic function $H^{0}$ and an odd superanalytic one $H^{1}$.
We introduce an odd derivative $D=\frac{\del}{\del \theta}+\theta\frac{\del}{\del z}$,
which is an square root of an even derivative in the sense that $D^{2}=\frac{\del}{\del z}$.
It can be verified that this odd derivative transforms under $H$ so that
\begin{equation}
	D=(D\tilde{\theta})\tilde{D}+(D\tilde{z}-\tilde{\theta}D\tilde{\theta})\frac{\del}{\del \tilde{z}}.
\end{equation}
We say that $H$ is superconformal if $D$ transforms homogeneously of degree one under $H$,
which is equivalent to the condition that $D\tilde{z}=\tilde{\theta}D\tilde{\theta}$.

To our purpose, it will be convenient to consider $z$ and $\theta$ as formal variables, even and odd, respectively,
and a superconformal map $H$ in a Laurent expanded form:
\begin{equation}
	H(z,\theta)=(f(z)+\theta \xi(z),\psi(z)+\theta g(z))\in ({\textstyle \bigwedge_{\ast}}[[z^{\pm}]][\theta])^{0}\oplus ({\textstyle \bigwedge_{\ast}}[[z^{\pm}]][\theta])^{1},
\end{equation}
where $f(z), g(z)\in \bigwedge_{\ast}^{0}[[z^{\pm}]]$ and $\xi(z),\psi(z)\in \bigwedge_{\ast}^{1}[[z^{\pm}]]$.
In these coordinates, we introduce vector fields
\begin{align}
	\mcal{L}^{(1)}_{j}&=-z^{j+1}\frac{\del}{\del z}-\left(\frac{j+1}{2}\right)\theta z^{j}\frac{\del}{\del\theta}, \\
	\mcal{G}_{j+\frac{1}{2}}&=-z^{j+1}\left(\frac{\del}{\del\theta}-\theta\frac{\del}{\del z}\right),
\end{align}
for $j\in\mbb{Z}$, which form a basis of an infinite dimensional Lie superalgebra denoted by $\mfrak{ns}_{1}^{0}$.
Here the subscript $1$ and the superscript $0$ express $\mcal{N}=1$ and the central charge $c=0$, respectively.
For a sequence $A_{j}\in \bigwedge_{\ast}^{0}$, $M_{j+\frac{1}{2}}\in\bigwedge_{\ast}^{1}$ for $j\in \mbb{Z}_{<0}$,
we define an operator $E_{A,M}$ on $(\bigwedge_{\ast}z[[z^{-1}]][\theta])^{0}\oplus (\bigwedge_{\ast}z[[z^{-1}]][\theta])^{1}$by
\begin{equation}
	\label{eq:N=1superconformal_vector_field}
	E_{A,M}=\exp\left(-\sum_{j\in\mbb{Z}_{<0}}\left(A_{j}\mcal{L}^{(1)}_{j}+M_{j+\frac{1}{2}}\mcal{G}_{j+\frac{1}{2}}\right)\right)
\end{equation}
Then this operator defines a formal superconformal map.
\begin{thm}
\label{thm:exp_superconformal_N=1}
The formal power series given by
\begin{equation}
	H(z,\theta)=(\tilde{z},\tilde{\theta})=E_{A,M}\cdot(z,\theta)
\end{equation}
defines a formal superconformal map, {\it i.e.}, it satisfies $D\tilde{z}=\tilde{\theta}D\tilde{\theta}$.
Here the operator $E_{A,M}$ acts on the variables componentwisely in the right hand side.
\end{thm}
\begin{proof}
We set
\begin{equation}
	T=-\sum_{j\in\mbb{Z}_{<0}}\left(A_{j}\mcal{L}^{(1)}_{j}+M_{j+\frac{1}{2}}\mcal{G}_{j+\frac{1}{2}}\right).
\end{equation}
Then we have $[D,T]=h(z,\theta)D$
with
\begin{equation}
	h(z,\theta)=\sum_{j\in\mbb{Z}_{<0}}\left(A_{j}\left(\frac{j+1}{2}\right)z^{j}+\theta M_{j+\frac{1}{2}}(j+1)z^{j}\right),
\end{equation}
which implies $e^{T+h(z,\theta)}D=De^{T}$.
We also have $e^{T+h(z,\theta)}\theta\cdot 1=(e^{T}\theta)(e^{T+h(z,\theta)}1)$.
Thus
\begin{equation}
	D\tilde{z}=e^{T+h(z,\theta)}\theta=(e^{T}\theta)(e^{T+h(z,\theta)}D\theta)=(e^{T}\theta)De^{T}\theta=\tilde{\theta}D\tilde{\theta}.
\end{equation}
\end{proof}

We denote the group generated by operators $E_{A,M}$ for various $A$, $M$ by $\mrm{SC}_{+}^{\mcal{N}=1}$,
which is the analogous object to $\mrm{Aut}_{+}\mcal{O}$ in Sect.\ref{sect:group_theor} in
case that $\mcal{N}=1$ supersymmetry is involved.

\subsection{Superconformal map of $\mcal{N}=2$}
The notion of superconformal map of $\mcal{N}=1$ recalled above can be naturally extended to the $\mcal{N}=2$ case.
Let $H$ be a superanalytic map of $(1,2)$-variables $H:\bigwedge_{\ast}^{0}\oplus (\bigwedge_{\ast}^{1})^{2}\to \bigwedge_{\ast}^{0}\oplus (\bigwedge_{\ast}^{1})^{2}$.
It sends $(z,\theta^{+},\theta^{-})\mapsto (\tilde{z},\tilde{\theta}^{+},\tilde{\theta}^{-})=(H^{0}(z,\theta^{+},\theta^{-}),H^{+}(z,\theta^{+},\theta^{-}),H^{-}(z,\theta^{+},\theta^{-}))$,
where $H^{0}$ is an even superanalytic function and $H^{\pm}$ are odd superanalytic functions in $(1,2)$-variables.
We consider two odd derivations $D^{\pm}=\frac{\del}{\del\theta^{\pm}}+\theta^{\mp}\frac{\del}{\del z}$,
which transform under $H$ as
\begin{equation}
	D^{\pm}=(D^{\pm}\tilde{\theta}^{\pm})\tilde{D}^{\pm}+(D^{\pm}\tilde{z}-\tilde{\theta}^{\mp}D^{\pm}\tilde{\theta}^{\pm})\frac{\del}{\del \tilde{z}}+(D^{\pm}\tilde{\theta}^{\mp})\frac{\del}{\del \tilde{\theta}^{\mp}}.
\end{equation}
The superanalytic map $H$ is said to be superconformal if $D^{\pm}$ transform homogeneously of degree one,
{\it i.e.}, if the following relations hold:
\begin{equation}
	D^{\pm}\tilde{z}=\tilde{\theta}^{\mp}D^{\pm}\tilde{\theta}^{\pm},\ \ \ D^{\pm}\tilde{\theta}^{\mp}=0.
\end{equation}
As we have done in the $\mcal{N}=1$ case, we consider $H$ in Laurent expended form
and treat variables $z,\theta^{+},\theta^{-}$ as formal ones.
Let us introduce following vector fields
\begin{align}
	\mcal{L}_{j}^{(2)}&=-\left(z^{j+1}\frac{\del}{\del z}+\left(\frac{j+1}{2}\right)z^{j}\left(\theta^{+}\frac{\del}{\del \theta^{+}}+\theta^{-}\frac{\del}{\del \theta^{-}}\right)\right), \\
	\mcal{J}_{j}&=-z^{j}\left(\theta^{+}\frac{\del}{\del\theta^{+}}-\theta^{-}\frac{\del}{\del\theta^{-}}\right),\\
	\mcal{G}_{j}^{\pm}&=-\left(z^{j+1}\left(\frac{\del}{\del \theta^{\pm}}-\theta^{\mp}\frac{\del}{\del z}\right)+(j+1)z^{j}\theta^{+}\theta^{-}\frac{\del}{\del\theta^{\pm}}\right),
\end{align}
for $j\in\mbb{Z}$, which form a basis of an infinite dimensional Lie superalgebra denoted by $\mfrak{ns}_{2}^{0}$.
Similarly to the $\mcal{N}=1$ case, the subscript $2$ and the superscript $0$ express $\mcal{N}=2$ and the central charge $c=0$, respectively.

For a sequence $A_{j}, B_{j}\in \bigwedge_{\ast}^{0}$, $M_{j+\frac{1}{2}}^{\pm}\in \bigwedge_{\ast}^{1}$,
we define an operator
\begin{equation}
	E_{A,B,M^{\pm}}=\exp\left(-\sum_{j\in\mbb{Z}_{<0}}\left(A_{j}\mcal{L}_{j}^{(2)}+B_{j}\mcal{J}_{j}+M_{j+\frac{1}{2}}^{+}\mcal{G}_{j+\frac{1}{2}}^{+}+M_{j+\frac{1}{2}}^{-}\mcal{G}_{j+\frac{1}{2}}^{-}\right)\right).
\end{equation}
The following theorem is proved in the same way as for Theorem \ref{thm:exp_superconformal_N=1}.
\begin{thm}
The superanalytic map given by
\begin{equation}
	H(z,\theta^{+},\theta^{-})=(\tilde{z},\tilde{\theta}^{+},\tilde{\theta}^{-})=E_{A,B,M^{\pm}}\cdot (z,\theta^{+},\theta^{-})
\end{equation}
is superconformal, {\it i.e.}, it satisfies $D^{\pm}\tilde{z}=\tilde{\theta}^{\mp}D^{\pm}\tilde{\theta}^{\pm}$ and $D^{\pm}\tilde{\theta}^{\mp}=0$.
\end{thm}

We denote the group generated by operators of the form $E_{A,B,M^{\pm}}$ for various coefficient data
$A, B ,M^{\pm}$ by $\mrm{SC}_{+}^{\mcal{N}=2}$.

\section{SLE on an $\mcal{N}=1$ super Riemann surface}
\label{sect:N=1}
In this section, we construct a generalization of SLE corresponding to the Neveu-Schwarz sector of the $\mcal{N}=1$ superconformal algebra $\mfrak{ns}_{1}$,
which is a central extension of $\mfrak{ns}_{1}^{0}$:
\begin{equation}
	\xymatrix{
		0 \ar[r] & \mbb{C}C \ar[r] & \mfrak{ns}_{1} \ar[r]^{\pi} & \mfrak{ns}_{1}^{0} \ar[r] & 0.
	}
\end{equation}
It is spanned by even generators $L_{n}$ ($n\in\mbb{Z}$), odd generators $G_{n+\frac{1}{2}}$ ($n\in\mbb{Z}$) and a central element $C$ with relations
\begin{align}
	[L_{m},L_{n}]&=(m-n)L_{m+n}+\frac{m^{3}-m}{12}\delta_{m+n,0}C, \\
	[G_{m+\frac{1}{2}},L_{n}]&=\left(m-\frac{n-1}{2}\right)G_{m+n+\frac{1}{2}}, \\
	[G_{m+\frac{1}{2}},G_{n-\frac{1}{2}}]&=2L_{m+n}+\frac{m^{2}+m}{3}\delta_{m+n,0}C.
\end{align}
The projection $\pi$ maps $L_{n}\mapsto \mcal{L}_{n}^{(1)}$ and $G_{n+\frac{1}{2}}\mapsto \mcal{G}_{n+\frac{1}{2}}$.

Let $\ket{c,h}$ be the highest weight vector of central charge $c$ and conformal weight $h$
that is annihilated by $L_{n}$ and $G_{n-\frac{1}{2}}$ for $n\in\mbb{Z}_{>0}$.
The Verma module induced from $\mbb{C}\ket{c,h}$ and its irreducible quotient are
denoted by $M(c,h)$ and $L(c,h)$, respectively.
The formal completion of the irreducible representation $L(c,h)$ is defined by
$\overline{L(c,h)}=\prod_{n\in\frac{1}{2}\mbb{Z}_{\ge 0}}L(c,h)_{h+n}$, where
$L(c,h)=\bigoplus_{n\in\frac{1}{2}\mbb{Z}_{\ge 0}}L(c,h)_{h+n}$ is the direct sum decomposition into
eigenspaces of $L_{0}$.
For an element $E_{A,M}\in\mrm{SC}_{+}^{\mcal{N}=1}$, we define an operator
\begin{equation}
	Q(E_{A,M})=\exp\left(-\sum_{j\in\mbb{Z}_{<0}}\left(A_{j}L_{j}+M_{j+\frac{1}{2}}G_{j+\frac{1}{2}}\right)\right),
\end{equation}
then $Q$ is a representation of $\mrm{SC}_{+}^{\mcal{N}=1}$ on $\overline{L(c,h)}\otimes\bigwedge_{\ast}$.

To construct a generalization of SLE, we need to fix a singular vector in $M(c,h)$.
We focus on a vector of the form
\begin{equation}
	\label{eq:singular_vector_N=1}
	\left(L_{-2}+a L_{-1}^{2}+b G_{-\frac{3}{2}}G_{-\frac{1}{2}}\right)\ket{c,h},
\end{equation}
which is verified to be a singular vector if $c=\frac{3}{2}\left(1-\frac{16}{3}h\right)$, $a=-b=-\frac{3}{4h}$.

Correspondingly to this singular vector, we consider a random process $H_{t}$ on $\mrm{SC}_{+}^{\mcal{N}=1}$ that satisfies
\begin{align}
	H_{t}^{-1}dH_{t}
		=&\left(-2\mcal{L}_{-2}^{1}+\frac{\kappa}{2}\left(\mcal{L}_{-1}^{(1)}\right)^{2}
		-\frac{\kappa}{4}\zeta_{2}\zeta_{1}\left(\mcal{G}_{-\frac{3}{2}}\mcal{G}_{-\frac{1}{2}}-\mcal{G}_{-\frac{1}{2}}\mcal{G}_{-\frac{3}{2}}\right)\right)dt \\
		&+\sqrt{\kappa}\mcal{L}_{-1}^{(1)}dB_{t}^{1}+\sqrt{\frac{\kappa}{2}}\left(\zeta_{1}\mcal{G}_{-\frac{1}{2}}+\zeta_{2}\mcal{G}_{-\frac{3}{2}}\right)dB_{t}^{2}, \notag
\end{align}
with the initial condition $H_{0}=\mrm{Id}$,
where $B_{t}^{1}$ and $B_{t}^{2}$ are mutually independent standard Brownian motions that start from the origin,
and $\zeta_{1}$ and $\zeta_{2}$ are two of generators of the coefficient algebra $\bigwedge_{\ast}$.
Then its value $H_{t}(z,\theta)=(H^{0}_{t}(z,\theta),H^{1}_{t}(z,\theta))$ satisfies
\begin{align}
	\label{eq:SLE_N=1_0}
	dH_{t}^{0}(z,\theta)
	=&\frac{2dt}{H^{0}(z,\theta)} - \sqrt{\kappa}dB_{t}^{1}
		+\sqrt{\frac{\kappa}{2}}\left(\zeta_{1}H_{t}^{1}(z,\theta)+\zeta_{2}\frac{H_{t}^{1}(z,\theta)}{H_{t}^{0}(z,\theta)}\right)dB_{t}^{2}, \\
	\label{eq:SLE_N=1_1}
	dH_{t}^{1}(z,\theta)
	=&\left(-\frac{H_{t}^{1}(z,\theta)}{H_{t}^{0}(z,\theta)}+\frac{\kappa}{4}\zeta_{2}\zeta_{1}\frac{H_{t}^{1}(z,\theta)}{H_{t}^{0}(z,\theta)^{2}}\right)dt
		+\sqrt{\frac{\kappa}{2}}\left(-\zeta_{1}-\frac{\zeta_{2}}{H_{t}^{0}(z,\theta)}\right)dB_{t}^{2},
\end{align}
with the initial conditions $H_{0}^{0}(z,\theta)=z$, $H_{0}^{1}(z,\theta)=\theta$,
the first of which can be regarded as a generalization of SLE in Eq.(\ref{eq:SLE_alter}).

To obtain local martingales associated with the stochastic differential equations (\ref{eq:SLE_N=1_0}) and (\ref{eq:SLE_N=1_1})
from a representation of $\mfrak{ns}_{2}$, we consider a random process $Q(H_{t})$ that
takes as its value operators on $\overline{L(c,h)}\otimes \bigwedge_{\ast}$.
It satisfies the stochastic differential equation
\begin{align}
	Q(H_{t})^{-1}dQ(H_{t})
		=&\left(-2L_{-2}+\frac{\kappa}{2}\left(L_{-1}\right)^{2}
		-\frac{\kappa}{4}\zeta_{2}\zeta_{1}\left(G_{-\frac{3}{2}}G_{-\frac{1}{2}}-G_{-\frac{1}{2}}G_{-\frac{3}{2}}\right)\right)dt \\
		&+\sqrt{\kappa}L_{-1}dB_{t}^{1}+\sqrt{\frac{\kappa}{2}}\left(\zeta_{1}G_{-\frac{1}{2}}+\zeta_{2}G_{-\frac{3}{2}}\right)dB_{t}^{2}, \notag
\end{align}
with the initial condition $Q(H_{0})=\mrm{Id}$.
Due to the singular vector in Eq.(\ref{eq:singular_vector_N=1}), the following quantity is a local martingale
with values in $\overline{L(c,h)}\otimes \bigwedge_{\ast>2}$:
\begin{equation}
	\int d\zeta_{1}d\zeta_{2}Q(H_{t})\ket{c_{\kappa},h_{\kappa}}\otimes (1+\zeta_{2}\zeta_{1}),
\end{equation}
where $c_{\kappa}=\frac{3}{2}-\frac{6(4-\kappa)}{\kappa}, h_{\kappa}=\frac{12-3\kappa}{4\kappa}$.
Here the integral in Grassmann variables $\int d\zeta_{1}d\zeta_{2}$ defines a map
$\bigwedge_{\ast}\to\bigwedge_{\ast>2}$, where $\bigwedge_{\ast>2}$ is the Grassmann algebra
generated by $\zeta_{3},\zeta_{4},\cdots$.

\section{SLE on an $\mcal{N}=2$ super Riemann surface}
\label{sect:N=2}
In this section, we construct a generalization of SLE corresponding to the Neveu-Schwarz sector of the $\mcal{N}=2$
superconformal algebra $\mfrak{ns}_{2}$.
It is a central extension
\begin{equation}
	\xymatrix{
		0 \ar[r] & \mbb{C}C \ar[r] & \mfrak{ns}_{2} \ar[r]^{\pi} & \mfrak{ns}_{2}^{0} \ar[r] & 0
	}
\end{equation}
of $\mfrak{ns}_{2}^{0}$ and spanned by even generators $L_{n}$, $J_{n}$, odd generators $G_{n+\frac{1}{2}}^{\pm}$ for $n\in\mbb{Z}$
and a central element $C$.
The Lie bracket among them is given by
\begin{align}
	[L_{m},L_{n}]&=(m-n)L_{m+n}+\frac{m^{3}-m}{12}\delta_{m+n,0}C, \\
	[J_{m},J_{n}]&=\frac{m}{3}\delta_{m+n,0}C, \\
	[L_{m},J_{n}]&=-nJ_{m+n}, \\
	[L_{m},G^{\pm}_{n+\frac{1}{2}}]&=\left(\frac{m-1}{2}-n\right)G^{\pm}_{m+n+\frac{1}{2}}, \\
	[J_{m},G^{\pm}_{n+\frac{1}{2}}]&=\pm G^{\pm}_{m+n+\frac{1}{2}}, \\
	[G^{\pm}_{m+\frac{1}{2}},G^{\pm}_{n+\frac{1}{2}}]&=0, \\
	[G^{+}_{m+\frac{1}{2}},G^{-}_{n+\frac{1}{2}}]&=2L_{m+n}+(m-n+1)J_{m+n}+\frac{m^{2}+m}{3}\delta_{m+n,0}C.
\end{align}
The projection maps $L_{n}\mapsto \mcal{L}_{n}^{(2)}$, $J_{n}\mapsto \mcal{J}_{n}$,
$G^{\pm}_{n+\frac{1}{2}}\mapsto \mcal{G}^{\pm}_{n+\frac{1}{2}}$.

Let $\ket{c,h,\alpha}$ be a highest weight vector that is a simultaneous eigenvector of
$C$, $L_{0}$ and $J_{0}$, corresponding to eigenvalues $c$, $h$, and $\alpha$, respectively,
and annihilated by $L_{n}$, $J_{n}$ and $G^{\pm}_{n-\frac{1}{2}}$ for $n\in\mbb{Z}_{>0}$.
The Verma module induced from $\mbb{C}\ket{c,h,\alpha}$ is denoted by $M(c,h,\alpha)$
and its irreducible quotient is denoted by $L(c,h,\alpha)$.
The formal completion $\overline{L(c,h,\alpha)}$ is defined in terms of the eigenspace decomposition
with respect to $L_{0}$.

For an element $E_{A,B,M^{\pm}}\in \mrm{SC}_{+}^{\mcal{N}=2}$, we define
\begin{equation}
	Q(E_{A,B,M^{\pm}})=\exp\left(-\sum_{j\in\mbb{Z}_{<0}}\left(A_{j}L_{j}+B_{j}J_{j}+M_{j+\frac{1}{2}}^{+}G_{j+\frac{1}{2}}^{+}+M_{j+\frac{1}{2}}^{-}G_{j+\frac{1}{2}}^{-}\right)\right)
\end{equation}
so that $Q$ becomes a representation of $\mrm{SC}_{+}^{\mcal{N}=2}$ on $\overline{L(c,h,\alpha)}\otimes \bigwedge_{\ast}$.

We consider a vector in the Verma module $M(c,h,\alpha)$ of the form
\begin{equation}
	\label{eq:singular_vector_N=2}
	\left(L_{-1}+\frac{1}{\alpha-1}G^{+}_{-\frac{1}{2}}G^{-}_{-\frac{1}{2}}+\frac{\alpha+1}{t}J_{-1}\right)\ket{c,h,\alpha}.
\end{equation}
It has been shown in \cite{Doerrzapf1995} that this is a singular vector if
\begin{equation}
	c=c(t)=3-3t,\ \ \ h=h(t,\alpha)=-\frac{1}{2}+\frac{3}{8t}-\frac{4\alpha^{2}-1}{t}.
\end{equation}
Correspondingly, we consider a random process $H_{t}$ on $\mrm{SC}_{+}^{\mcal{N}=2}$ that satisfies
\begin{align}
	H_{t}^{-1}dH_{t}=
	&\left(\mcal{L}^{(2)}_{-1}+a\mcal{J}_{-1}+\frac{\kappa}{2}\zeta_{1}\zeta_{2}\left(\mcal{G}^{+}_{-\frac{1}{2}}\mcal{G}^{-}_{-\frac{1}{2}}-\mcal{G}^{-}_{-\frac{1}{2}}\mcal{G}^{+}_{-\frac{1}{2}}\right)\right)dt \\
	&+\sqrt{\kappa}\left(\zeta_{1}\mcal{G}^{+}_{-\frac{1}{2}}+\zeta_{2}\mcal{G}^{-}_{-\frac{1}{2}}\right)dB_{t}, \notag
\end{align}
with the initial condition $H_{0}=\mrm{Id}$,
where $B_{t}$ is the standard Brownian motion starting from the origin and parameters are set as $a=\frac{(\alpha-1)^{2}}{t\alpha}$ and $\kappa=\frac{1}{\alpha}$.
The stochastic differential equation for the value
\begin{equation}
	H_{t}(z,\theta^{+},\theta^{-})=(H_{t}^{0} (z,\theta^{+},\theta^{-}),H_{t}^{+} (z,\theta^{+},\theta^{-}),H_{t}^{-} (z,\theta^{+},\theta^{-}))
\end{equation}
can be also written down as
\begin{align}
	\label{eq:SLE_N=2_0}
	dH_{t}^{0}(z,\theta^{+},\theta^{-})&=-dt +\sqrt{\kappa}(\zeta_{1}H_{t}^{+}(z,\theta^{+},\theta^{-})+\zeta_{2}H_{t}^{-}(z,\theta^{+},\theta^{-}))dB_{t},\\
	\label{eq:SLE_N=2_+}
	dH_{t}^{+}(z,\theta^{+},\theta^{-})&=-a\frac{H_{t}^{+}(z,\theta^{+},\theta^{-})}{H_{t}^{0}(z,\theta^{+},\theta^{-})}dt -\sqrt{\kappa}\zeta_{2}dB_{t},\\
	\label{eq:SLE_N=2_-}
	dH_{t}^{-}(z,\theta^{+},\theta^{-})&=a\frac{H_{t}^{-}(z,\theta^{+},\theta^{-})}{H_{t}^{0}(z,\theta^{+},\theta^{-})}dt -\sqrt{\kappa}\zeta_{1}dB_{t},
\end{align}
with the initial conditions $H_{0}^{0}(z,\theta^{+},\theta^{-})=z$, $H_{0}^{+}(z,\theta^{+},\theta^{-})=\theta^{+}$,
$H_{0}^{-}(z,\theta^{+},\theta^{-})=\theta^{-}$.

To derive a vector-valued local martingale associated with the stochastic differential equations
(\ref{eq:SLE_N=2_0}), (\ref{eq:SLE_N=2_+}) and (\ref{eq:SLE_N=2_-}),
we consider the random process $Q(H_{t})$,
of which value are operators on $\overline{L(c,h,\alpha)}\otimes \bigwedge_{\ast}$.
It satisfies the following stochastic differential equation
\begin{align}
	Q(H_{t})^{-1}dQ(H_{t})=
	&\left(L_{-1}+aJ_{-1}+\frac{\kappa}{2}\zeta_{1}\zeta_{2}(G^{+}_{-\frac{1}{2}}G^{-}_{-\frac{1}{2}}-G^{-}_{-\frac{1}{2}}G^{+}_{-\frac{1}{2}})\right)dt \\
	&+\sqrt{\kappa}(\zeta_{1}G^{+}_{-\frac{1}{2}}+\zeta_{2}G^{-}_{-\frac{1}{2}})dB_{t}, \notag
\end{align}
with the initial condition $Q(H_{0})=\mrm{Id}$.
Then it can be verified that
\begin{equation}
	\int d\zeta_{2}d\zeta_{1}Q(H_{t})\ket{c(t),h(t,\alpha),\alpha}\otimes (1+\zeta_{1}\zeta_{2})	
\end{equation}
is a local martingale
because of the singular vector in Eq.(\ref{eq:singular_vector_N=2}).

\section{Discussion}
\label{sect:discussion}
We have proposed generalizations of SLE that are associated with
superconformal algebras of $\mcal{N}=1$ (Eq.(\ref{eq:SLE_N=1_0}), (\ref{eq:SLE_N=1_1}))
and $\mcal{N}=2$ (Eq.(\ref{eq:SLE_N=2_0}), (\ref{eq:SLE_N=2_+}), (\ref{eq:SLE_N=2_-})).
These stochastic differential equations are derived from random processes
on infinite dimensional Lie groups of superconformal maps.
Such construction allows one to obtain local martingales
associated with the solutions from a representation of
a superconformal algebra after certain integral over Grassmann variables,
which was also presented in this paper following the prescription used in our previous work \cite{SK2018b}.
Though our construction assumes a specific form of singular vectors,
it can be extended to another singular vector as long as it is obtained
by applying an operator, at most quadratic in generators, to a highest weight vector.

A generalization of SLE corresponding to the $\mcal{N}=1$ superconformal algebra
has been already considered in \cite{Rasmussen2004a} for the Neveu-Schwarz case
and in \cite{NagiRasmussen2005} for the Ramond case.
Compared to these works, in which the authors focused on an odd singular vector,
this paper treats an even singular vector,
which requires one to integrate out some Grassmann variables to obtain local martingales.
Corresponding to this difference in approach, our stochastic differential equations (\ref{eq:SLE_N=1_0}) and (\ref{eq:SLE_N=1_1})
are different from ones discovered in \cite{Rasmussen2004a},
but ours seems to be a more natural candidate for a generalization of SLE with supersymmetry.

There are several future directions concerning SLE associated with superconformal algebras.
Though, in the present paper, we focused on the Neveu-Schwarz sector of superconformal algebras,
our construction will also be applied to the Ramond sector.
Probably the most important one of future directions is to construct a multiple version of SLE with supersymmetry presented in this paper,
which will allow one to understand more deeply SLE with supersymmetry in connection with CFT.
Such a work will be supported by better understanding of the partition function of SLE with supersymmetry.
Related to this, we also mention SLE as a random process on the moduli space of Riemann surfaces
established in \cite{FriedrichKalkkinen2004,Kontsevich2003}.
We suspect that SLE with supersymmetry can also be regarded as a random process on
the moduli space of super Riemann surfaces.
Such an understanding of SLE with supersymmetry will allow one to realize it
on more general super Riemann surfaces than one of genus 0.

\section*{Ackowledgement}
This work was supported by a Grant-in-Aid for JSPS Fellows (Grant No. 17J09658).












\addcontentsline{toc}{chapter}{Bibliography}
\bibliographystyle{alpha}
\bibliography{sle_cft}



\end{document}